%% file: lowerbounds-arxiv.tex
\documentclass[12pt]{article}

\usepackage{amsmath,  
            mathrsfs} 

\usepackage{tikz}
\usetikzlibrary{arrows,snakes,backgrounds}
\tikzstyle{vertex} = [circle,fill=black!0,minimum size=3pt,inner sep=0pt]

\usepackage[numbers]{natbib}
\usepackage{graphicx}
\include{newtheorems-en}

\newcommand{\ie}{i.e.\ }
\newcommand{\eg}{e.g.\ }
\newcommand{\resp}{resp.\ }
\title{Lower bounding edit distances between permutations\thanks{A significant portion of this work previously appeared in the Proceedings of the Sixteenth Annual European Symposium on Algorithms (ESA)~\cite{labarre-edit}.}}
\author{Anthony Labarre}

\begin{document}

\maketitle

\begin{abstract}
A number of fields, including the study of genome rearrangements and the design of interconnection networks, deal with the connected problems of sorting permutations in ``as few moves as possible'', using a given set of allowed operations, or computing the number of moves the sorting process requires, often referred to as the \emph{distance} of the permutation. These operations often act on just one or two segments of the permutation, \eg by reversing one segment or exchanging two segments. The \emph{cycle graph} of the permutation to sort is a fundamental tool in the theory of genome rearrangements, and has proved useful in settling the complexity of many variants of the above problems. In this paper, we present an algebraic reinterpretation of the cycle graph of a permutation $\pi$ as an even permutation $\overline{\pi}$, and show how to reformulate our sorting problems in terms of particular factorisations of the latter permutation. Using our framework, we recover known results in a simple and unified way, and obtain a new lower bound on the \emph{prefix transposition distance} (where a \emph{prefix transposition} displaces the initial segment of a permutation), which is shown to outperform previous results. Moreover, we use our approach to improve the best known lower bound on the \emph{prefix transposition diameter} from $2n/3$ to $\left\lfloor 3n/4\right \rfloor$, and investigate a few relations between some statistics on $\pi$ and $\overline{\pi}$.
\end{abstract}


\section{Introduction}

Given a set $S$ of allowed operations and two permutations $\pi$ and $\sigma$ of $\{1,2,\ldots,n\}$, we study the related problems of computing, on the one hand, a sequence of elements of $S$ of minimum length that transforms $\pi$ into $\sigma$, and on the other hand, computing the length of such a sequence, referred to as the \emph{distance} between $\pi$ and $\sigma$. The operations in $S$ usually yield an \emph{edit distance} $d_S(\cdot, \cdot)$ with the property that $d_S(\pi, \sigma)=d_S(\sigma^{-1}\circ\pi, \iota)$ for any two permutations $\pi$ and $\sigma$ of the same set, where $\iota$ is the  \emph{identity permutation} $\langle 1\ 2\ \cdots\ n\rangle$. This property allows us to restrict our attention to \emph{sorting permutations} using a minimum number of operations from $S$, or to computing the distance of a given permutation to the identity permutation rather than to another arbitrary permutation.  Two areas in which these questions have applications are the fields of \emph{genome rearrangements} and \emph{interconnection network design}, which we briefly review below.

In genome rearrangements (see \citet{fertin-combinatorics} for a survey), the permutation to sort represents an ordering of genes in a given genome, and the allowed operations model \emph{mutations} that are known to actually occur in evolution. Rearrangements studied in that context include \emph{reversals}~\cite{kececioglu-exact}, which reverse a segment of the permutation, \emph{block-interchanges}~\cite{christie-block}, which exchange two not necessarily contiguous segments, and \emph{transpositions}~\cite{bafna-transpositions}, which displace a block of contiguous elements. Those seemingly easy problems turn out to be more challenging than they might appear at first: although a polynomial-time algorithm is known for sorting by block-interchanges or computing the associated distance~\cite{christie-block}, the same problems were shown to be NP-hard for reversals~\cite{caprara-sorting}, and more recently for transpositions~\cite{Bulteau2011}.

In interconnection network design (see \citet{lakshmivarahan-symmetry} for a thorough survey), permutations stand \eg for processors, or other devices to be connected, and form the vertex set of a graph whose edges correspond to physical connections between two devices. One wants to build a graph with small degree and small \emph{diameter}, among other desirable properties. \citeauthor{akers-group}'s landmark paper~\cite{akers-group} proposed the idea of choosing a set $S$ that \emph{generates} all permutations of $\{1,2,\ldots,n\}$, and to use the corresponding \emph{Cayley graph}, whose vertex set is the set of all permutations and whose edges connect any two permutations that can be obtained from one another by applying a transformation from $S$, as an interconnection network. 
In that setting, sorting algorithms for permutations correspond to \emph{routing algorithms} for the corresponding networks, since a sequence of elements of $S$ transforming $\pi$ into $\sigma$ corresponds to a path of the same length in the network. 
Two kinds of operations that received a lot of attention in that context are \emph{prefix reversals}~\cite{gates-bounds}, which reverse the initial segment of the permutation, and \emph{prefix exchanges}~\cite{akers-star}, which swap the first element of the permutation with another element. Those operations gave birth to the \emph{pancake network} and \emph{star graph} topologies, respectively, which are extensively studied models in that field. We also mention \emph{prefix transpositions}, which displace the initial segment of the permutation, and were introduced by \citet{dias-prefix} in the context of genome rearrangements in the hope that their study would shed light and give insight on the challenging problem of sorting by transpositions. Those more restricted versions of operations studied in the context of genome rearrangements do not lead to problems simpler than their unrestricted counterparts: the sorting and distance computation problems related to prefix exchanges can be solved in polynomial time~\cite{akers-star}, but the complexity of those problems in the case of prefix transpositions is open, and the problem of sorting by prefix reversals has only recently been showed to be NP-hard~\cite{Bulteau2011-pancakes}, more than thirty years after the first works on the subject~\cite{gates-bounds,gyori-stack}.

The \emph{cycle graph} of a permutation is a ubiquitous structure in the field of genome rearrangements, and has proved useful in resolving many questions related to the problems discussed in the above paragraphs. In this paper, we present a new way of encoding the cycle graph of a permutation $\pi$ as an even permutation $\overline{\pi}$, inspired by a previous work of ours~\cite{doignon-hultman}, and show how to reformulate \emph{any} sorting problem of the form described above in terms of particular factorisations of the latter permutation. We first illustrate the power of our framework by recovering known lower bounds on the block-interchange and transposition distances in a simple and unified way, and then use it to prove a new lower bound on the prefix transposition distance. We prove that our lower bound always outperforms that obtained by \citet{dias-prefix}, and show experimentally that it is a significant improvement over both that result and the only other known lower bound proved by \citet{chitturi-bounding}. We then use this new result to improve the previously best known lower bound on the maximal value of the prefix transposition distance from $2n/3$ to $\left\lfloor 3n/4 \right \rfloor$. Finally, we examine some further properties of the model, and establish connections between statistics on $\pi$ and $\overline{\pi}$.


\section{Notation and definitions}

\subsection{Permutations and conjugacy classes}

Let us start with a quick reminder of basic notions on permutations (for details, see \eg \citet{wielandt-finite}).

\begin{definition}\label{def:permutation}
A \emph{permutation} of 
a set $\Omega$ 
is a bijective application of 
$\Omega$ 
onto itself.
\end{definition}

It is convenient to set $\Omega=\{1, 2, \ldots, n\}$, and we will follow this convention here, although we will also sometimes use the set $\{0, 1, 2, \ldots, n\}$. The \emph{symmetric group} $S_n$ is the set of all permutations of a set of $n$ elements, together with the usual function composition $\circ$, applied from right to left. 
Permutations are denoted by lower case Greek letters, and we will follow the convention of shortening the traditional two-row notation
$$
\pi=\left\langle
\begin{array}{cccc}
1& 2& \cdots& n\\
\pi_1& \pi_2& \cdots& \pi_n
\end{array}
\right\rangle
$$
by keeping only the second row, \ie $\pi=\langle\pi_1\ \pi_2\ \cdots\ \pi_n\rangle$, where  $\pi_i=\pi(i)$. 

\begin{definition}
The \emph{graph} $\Gamma(\pi)$ of the permutation $\pi$ in $S_n$ is the directed graph with ordered vertex set $(\pi_1, \pi_2, \ldots, \pi_n)$ and arc set $\{(i,j)\ |\ \pi_i=j, 1\leq i\leq n\}$.
\end{definition}

Definition~\ref{def:permutation} implies that $\Gamma(\pi)$ decomposes in a single way into disjoint cycles (up to the ordering of cycles and of elements within each cycle), leading to another notation for $\pi$ based on its \emph{disjoint cycle decomposition}.  For instance, when $\pi=\langle 4\ 1\ 6\ 2\ 5\ 7\ 3\rangle$, the disjoint cycle notation is $\pi=(1,4,2)(3,6,7)(5)$ (notice the parentheses and the commas). 

\begin{definition}\label{def:cycle-length}
The \emph{length} of a cycle in a graph is the number of vertices it contains, and a \emph{$k$-cycle} is a cycle of length $k$.
\end{definition}

The number of cycles in a graph $G$ will be denoted by $c(G)$, and the number of cycles of length $k$ will be denoted by $c_k(G)$. We will also distinguish between cycles of odd (\resp even) length, denoting the number of such cycles in $G$ using $c_{odd}(G)$ (\resp $c_{even}(G)$).
It is common practice to omit $1$-cycles in the cycle decomposition of (the graph of) a permutation, and to call that permutation a $k$-cycle if the resulting decomposition consists of a single cycle of length $k>1$. Cycles of length $1$ in the disjoint cycle decomposition of a permutation are referred to as \emph{fixed points}.

\begin{definition}\label{def:even-permutation}
A permutation $\pi$ is \emph{even} if the number of even cycles in $\Gamma(\pi)$ is even  or, equivalently, if it can be expressed as a product of an even number of $2$-cycles. 
\end{definition}

The \emph{alternating group} $A_n$ is the subgroup of $S_n$ formed by the set of all even permutations, together with $\circ$. 
The following notion will be central to this work.

\begin{definition}
The \emph{conjugate} of a permutation $\pi$ by a permutation $\sigma$, both in $S_n$, is the permutation $\pi^{\sigma}=\sigma\circ\pi\circ\sigma^{-1}$, and can be obtained by replacing every element $i$ in the disjoint cycle decomposition of $\pi$ with $\sigma_i$. All permutations in $S_n$ that have the same disjoint cycle decomposition form a \emph{conjugacy class} (of $S_n$).
\end{definition}

\subsection{Generating sets and edit distances}\label{sec:intro-gr-prefix}

We are interested in distances between permutations based on operations that can themselves be modelled as permutations. More formally, given a subset $S$ of $S_n$ and two permutations $\pi$ and $\sigma$ in $S_n$, we have two goals:
\begin{enumerate}
 \item to find a sequence of elements $s_1, s_2, \ldots, s_t$ from $S$ whose cardinality is minimum and whose product transforms $\pi$ into $\sigma$ (or conversely, $\sigma$ into $\pi$): $$\pi\circ s_1\circ s_2\circ\cdots\circ s_t=\sigma.$$
 \item to find the length of such a sequence, called the \emph{$S$ distance} between $\pi$ and $\sigma$. Distances whose definition is based on a set of allowed operations as described above are often referred to as \emph{edit distances}.
\end{enumerate}

Note that $S$ must be \emph{symmetric}, \ie $s\in S$ if and only if $s^{-1}\in S$, for the corresponding distance to satisfy the symmetry axiom. An immediate corollary of this property is that for any $\pi$ in $S_n$, we have $d(\pi, \iota)=d(\pi^{-1}, \iota)$. For any two permutations of the same set to be a finite distance apart, $S$ must also satisfy the following property.

\begin{definition}\label{def:generating-set}
A set $S\subset S_n$ is said to \emph{generate} $S_n$, or to be a \emph{generating set} of $S_n$, if every element of $S_n$ can be expressed as the product of a finite number of elements of $S$. We call the elements of $S$ \emph{generators} of $S_n$.
\end{definition}

Moreover, all generating sets we will consider in this paper yield distances that satisfy the following property.

\begin{definition}\label{def:left-invariant-distance}
A distance $d$ on $S_n$ is \emph{left-invariant} if for all $\pi$, $\sigma$, $\tau$ in $S_n$, we have: $d(\pi, \sigma) = d(\tau\circ\pi, \tau\circ\sigma).$
\end{definition}

Intuitively, left-invariance models the fact that, given any two permutations $\pi$ and $\sigma$ to be transformed into one another, we can rename the elements of either permutation as we wish without changing the value of the distance between both permutations, as long as we renumber the elements of the other permutation accordingly. Since we will most of the time be considering the distance between a permutation $\pi$ and the identity permutation $\iota$, we will often abbreviate $d(\pi,\iota)$ to $d(\pi)$. 

It can be easily seen that both problems mentioned at the beginning of this section can be reformulated in terms of finding a minimum-length factorisation of $\pi$ that consists only of elements of $S$, since
$$\pi\circ s_1\circ s_2\circ \cdots\circ s_t=\iota\Leftrightarrow\pi=s_t^{-1}\circ s_{t-1}^{-1}\circ\cdots\circ s_1^{-1}$$
and $S$ is symmetric.
Finally, another parameter of interest in the study of those distances is the largest value they can reach.
\begin{definition}
The \emph{diameter} of a set $U$ under a distance $d$ is $\max_{s,t\in U}d(s,t)$.
\end{definition}

\subsection{Genome rearrangements and the cycle graph}

We recall here a few operations that are commonly used in the fields of genome rearrangements and interconnection network design to build generating sets of $S_n$.

\begin{definition}\label{def:block-interchange}
\cite{christie-block} 
The \emph{block-interchange} $\beta(i, j, k, l)$ with $1\leq i<j\leq k<l\leq n+1$ 
is the permutation that exchanges the closed intervals determined respectively by $i$ and $j-1$ and by $k$ and $l-1$:
$$
\left\langle
\begin{array}{l}
1\ \cdots\ i-1\ \fbox{$i\ \cdots\ j-1$}\ j\ j+1\ \cdots\ k-1\ \fbox{$k\ \cdots\ l-1$}\ l\ l+1\ \cdots\ n \\
\raisebox{-.05in}{$1\ \cdots\ i-1\ \fbox{$k\ \cdots\ l-1$}\ j\ j+1\ \cdots\ k-1\ \fbox{$i\ \cdots\ j-1$}\ l\ l+1\ \cdots\ n$} \\
\end{array}
\right\rangle.
$$ 
\end{definition}

\noindent Two particular cases of block-interchanges are of interest: 
\begin{enumerate}
 \item 
when $j=k$, the resulting operation exchanges two adjacent intervals, and is called a \emph{transposition}~\cite{bafna-transpositions}, denoted by $\tau(i, j, l)$; 
\item when $j=i+1$ and $l=k+1$, the resulting operation swaps two not necessarily adjacent elements in respective positions $i$ and $k$, and is called an \emph{exchange}, denoted by $\varepsilon(i, k)$. 
\end{enumerate}

We use the notation $bid(\pi)$, $td(\pi)$ and $exc(\pi)$ for the block-interchange distance, the transposition distance, and the exchange distance of $\pi$, respectively. The operations we described above can be further restricted by setting $i=1$ in their definition, thereby transforming them into so-called ``prefix rearrangements''. 
The corresponding ``prefix distances'' are defined in an analogous manner, with the additional restriction that all operations must act on the initial segment of the permutation. We denote $ptd(\pi)$ and $pexc(\pi)$ the \emph{prefix transposition distance} and \emph{prefix exchange distance} of $\pi$, respectively. While sorting by transpositions is NP-hard~\cite{Bulteau2011} and the computational complexity of sorting by prefix transpositions is unknown, polynomial-time algorithms exist for sorting by block-interchanges~\cite{christie-block}, exchanges~\cite{jerrum-complexity} or prefix exchanges~\cite{akers-star}, as well as formulas for computing the associated distances.

We will have more to say about sorting by transpositions and sorting by block-interchanges in Section~\ref{sec:recovery}, where we will give simple proofs of lower bounds on the two corresponding distances, as well as about sorting by prefix transpositions in Section~\ref{sec:sbpt}, where we will prove new and improved lower bounds on the associated distance and diameter. Meanwhile, we conclude this section with the following traditional tool introduced by \citet{bafna-transpositions}, which has proved most useful in the study of genome rearrangements. 

\begin{definition}\label{def:cycle-graph}
The \emph{cycle graph} of a permutation $\pi$ in $S_n$ is the bicoloured directed graph $G(\pi)$, whose vertex set $(\pi_0=0, \pi_1, \ldots, \pi_n)$ is ordered by positions, and whose arc set consists of:
\begin{itemize}
\item \emph{black} arcs $\{(\pi_i, \pi_{i-1})\ |\ 1\leq i\leq n\}\cup\{(\pi_0,\pi_n)\}$;
\item \emph{grey} arcs $\{(\pi_i, \pi_i+1)\ |\ 0\leq i\leq n\}\cup\{(n,0)\}$.
\end{itemize} 
\end{definition}

The arc set of $G(\pi)$ 
decomposes in a single way into arc-disjoint \emph{alternating cycles}, \ie cycles that alternate black and grey arcs.
The \emph{length} of an alternating cycle in $G(\pi)$ is the number of black arcs it contains, and a \emph{$k$-cycle} in $G(\pi)$ is an alternating cycle of length $k$ (note that this differs from Definition~\ref{def:cycle-length}). Figure~\ref{fig:modified-cycle-graph-example} shows an example of a cycle graph, together with its decomposition into a $5$-cycle and a $3$-cycle.

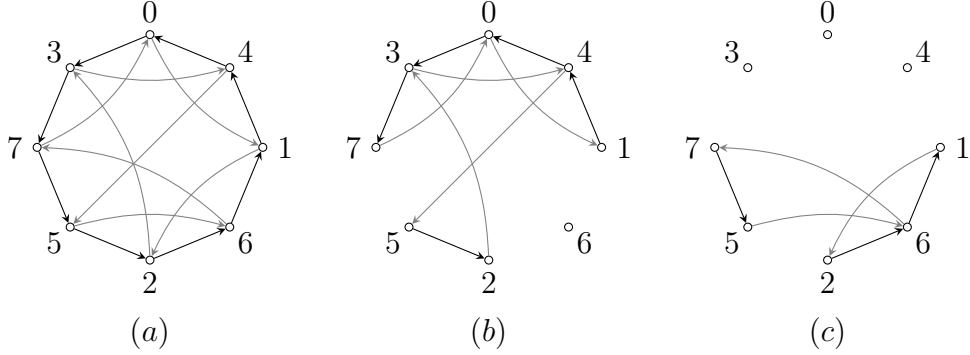
\begin{figure}[htbp]
\begin{center}
\begin{tabular}{ccc}
\begin{tikzpicture}[scale=0.75,>=stealth]
  \path (0*45:2.0cm) node [vertex,draw,circle] (v1) {};
  \path (1*45:2.0cm) node [vertex,draw,circle] (v2) {};
  \path (2*45:2.0cm) node [vertex,draw,circle] (v3) {};
  \path (3*45:2.0cm) node [vertex,draw,circle] (v4) {};
  \path (4*45:2.0cm) node [vertex,draw,circle] (v5) {};
  \path (5*45:2.0cm) node [vertex,draw,circle] (v6) {};
  \path (6*45:2.0cm) node [vertex,draw,circle] (v7) {};
  \path (7*45:2.0cm) node [vertex,draw,circle] (v8) {};
  \path (0*45:2.4cm) node {$1$};
  \path (1*45:2.4cm) node {$4$};
  \path (2*45:2.4cm) node {$0$};
  \path (3*45:2.4cm) node {$3$};
  \path (4*45:2.4cm) node {$7$};
  \path (5*45:2.4cm) node {$5$};
  \path (6*45:2.4cm) node {$2$};
  \path (7*45:2.4cm) node {$6$};
  \draw[->] (v1) -- (v2);
  \draw[->] (v2) -- (v3);
  \draw[->] (v3) -- (v4);
  \draw[->] (v4) -- (v5);
  \draw[->] (v5) -- (v6);
  \draw[->] (v6) -- (v7);
  \draw[->] (v7) -- (v8);
  \draw[->] (v8) -- (v1);
  \draw[gray,->] (v1) to [bend right=20] (v7);
  \draw[gray,->] (v7) to [bend right=20] (v4);
  \draw[gray,->] (v5) to [bend right=20] (v3);
  \draw[gray,->] (v2) -- (v6);
  \draw[gray,->] (v8) to [bend right=20] (v5);
  \draw[gray,->] (v4) to [bend right=15] (v2);
  \draw[gray,->] (v6) to [bend left=15] (v8);
  \draw[gray,->] (v3) to [bend right=20] (v1);
\end{tikzpicture}
&
\begin{tikzpicture}[scale=0.75,>=stealth]
  \path (0*45:2.0cm) node [vertex,draw,circle] (v1) {};
  \path (1*45:2.0cm) node [vertex,draw,circle] (v2) {};
  \path (2*45:2.0cm) node [vertex,draw,circle] (v3) {};
  \path (3*45:2.0cm) node [vertex,draw,circle] (v4) {};
  \path (4*45:2.0cm) node [vertex,draw,circle] (v5) {};
  \path (5*45:2.0cm) node [vertex,draw,circle] (v6) {};
  \path (6*45:2.0cm) node [vertex,draw,circle] (v7) {};
  \path (7*45:2.0cm) node [vertex,draw,circle] (v8) {};
  \path (0*45:2.4cm) node {$1$};
  \path (1*45:2.4cm) node {$4$};
  \path (2*45:2.4cm) node {$0$};
  \path (3*45:2.4cm) node {$3$};
  \path (4*45:2.4cm) node {$7$};
  \path (5*45:2.4cm) node {$5$};
  \path (6*45:2.4cm) node {$2$};
  \path (7*45:2.4cm) node {$6$};
  \draw[->] (v1) -- (v2);
  \draw[->] (v2) -- (v3);
  \draw[->] (v3) -- (v4);
  \draw[->] (v4) -- (v5);
  \draw[->] (v6) -- (v7);

  \draw[gray,->] (v7) to [bend right=20] (v4);
  \draw[gray,->] (v5) to [bend right=20] (v3);
  \draw[gray,->] (v2) -- (v6);
  \draw[gray,->] (v4) to [bend right=15] (v2);
  \draw[gray,->] (v3) to [bend right=20] (v1);
\end{tikzpicture}
&
\begin{tikzpicture}[scale=0.75,>=stealth]
  \path (0*45:2.0cm) node [vertex,draw,circle] (v1) {};
  \path (1*45:2.0cm) node [vertex,draw,circle] (v2) {};
  \path (2*45:2.0cm) node [vertex,draw,circle] (v3) {};
  \path (3*45:2.0cm) node [vertex,draw,circle] (v4) {};
  \path (4*45:2.0cm) node [vertex,draw,circle] (v5) {};
  \path (5*45:2.0cm) node [vertex,draw,circle] (v6) {};
  \path (6*45:2.0cm) node [vertex,draw,circle] (v7) {};
  \path (7*45:2.0cm) node [vertex,draw,circle] (v8) {};
  \path (0*45:2.4cm) node {$1$};
  \path (1*45:2.4cm) node {$4$};
  \path (2*45:2.4cm) node {$0$};
  \path (3*45:2.4cm) node {$3$};
  \path (4*45:2.4cm) node {$7$};
  \path (5*45:2.4cm) node {$5$};
  \path (6*45:2.4cm) node {$2$};
  \path (7*45:2.4cm) node {$6$};
  \draw[->] (v5) -- (v6);
  \draw[->] (v7) -- (v8);
  \draw[->] (v8) -- (v1);
  \draw[gray,->] (v1) to [bend right=20] (v7);
  \draw[gray,->] (v8) to [bend right=20] (v5);
  \draw[gray,->] (v6) to [bend left=15] (v8);

\end{tikzpicture}
\\
$(a)$ & $(b)$ & $(c)$ \\
\end{tabular}
\end{center}
\caption{$(a)$ The cycle graph of $\langle 4\ 1\ 6\ 2\ 5\ 7\ 3\rangle$; $(b),(c)$ its decomposition into two alternating cycles.}
\label{fig:modified-cycle-graph-example}
\end{figure}


\section{A general lower bounding technique}\label{sec:general-lower-bounding-technique}

We now present a framework for obtaining lower bounds on edit distances between permutations in a simple and unified way. To that end, we adapt a bijection previously introduced by
\citet{doignon-hultman}:
\begin{equation}\label{eq-alpha}
f:S_n\rightarrow A_{n+1}:\pi\mapsto\overline{\pi} =  (0,1,2,\ldots,n)\circ (0,\pi_n,\pi_{n-1},\ldots,\pi_1),
\end{equation}
which in particular maps $\iota$ onto $\overline{\iota}=\langle 0\ 1\ 2\ \cdots\ n\rangle$. That mapping allows us to encode the structure of a cycle graph $G(\pi)$ using an even permutation $\overline{\pi}$ in an intuitive way, which corresponds to decomposing the cycle graph into the product of two ``monochromatic cycles'', namely, the cycle made of all black arcs (\ie $(0,\pi_n,\pi_{n-1},\ldots,\pi_1)$) and the cycle made of all grey arcs (\ie $(0,1,2,\ldots,n)$). The construction is perhaps best understood using an example: let $\pi=\langle 4\ 1\ 6\ 2\ 5\ 7\ 3\rangle$, whose cycle graph is depicted in Figure~\ref{fig:modified-cycle-graph-example}$(a)$. Then 
$$\overline\pi=(0,1,2,3,4,5,6,7)\circ(0,3,7,5,2,6,1,4)=(0,4,1,5,3)(2,7,6),$$ 
and the two disjoint cycles of $\overline{\pi}$ correspond to the two alternating cycles of $G(\pi)$, whose elements they list in the order they are encountered (up to rotation); indeed:
\begin{enumerate}
 \item the first cycle of $G(\pi)$ (Figure~\ref{fig:modified-cycle-graph-example}$(b)$) starts with $0$, then visits $4$ after following a black-grey path (\ie a black arc followed by a grey arc), then visits $1$ after following a black-grey path, and in the same way visits $5$ and $3$ before coming back to $0$, which corresponds to the first cycle of $\overline{\pi}$;
 \item the second cycle of $G(\pi)$ (Figure~\ref{fig:modified-cycle-graph-example}$(c)$) starts with $2$, then visits $7$ after following a black-grey path, and in the same way visits $6$ before coming back to $2$, which corresponds to the second cycle of $\overline{\pi}$.
\end{enumerate}

Note that the order in which we decide to follow arcs (first a black arc and then a grey arc) is given by the order in which the two cycles are multiplied. An alternative definition\footnote{This is actually the definition we used in the conference version of this paper~\cite{labarre-edit}.} of $\overline{\pi}$ could therefore have been $(0,\pi_n,\pi_{n-1},\ldots,\pi_1)\circ(0,1,2,\ldots,n)$, which can be seen to be equivalent to our definition when conjugated by $(0,n,n-1,\ldots,1)$, and whose cycles are interpreted exactly as above, with the modification that grey arcs are followed first. 
Consequently, speaking about cycles of $\overline{\pi}$, of $\Gamma(\overline{\pi})$ or of $G(\pi)$ is equivalent. We will now demonstrate how $f(\cdot)$ can be used to obtain results on the sorting and distance computation problems we discussed in Section~\ref{sec:intro-gr-prefix}. The following lemma expresses how the action of \emph{any} rearrangement operation $\sigma$ on $\pi$ is translated on $\overline{\pi}$. We will find it convenient to identify permutations in $S_n$ with their extended versions in $S_{n+1}$ (\ie we identify $\pi$ with $\langle 0\ \pi_1\ \pi_2\ \cdots\ \pi_n\rangle$). This allows us to express any permutation $\pi$ in $S_n$ as follows:
\begin{equation}\label{eqn:pi-from-overline-pi}
\pi = (0,\pi_n,\pi_{n-1},\ldots,\pi_1)\circ\pi\circ(0,1,2,\ldots,n).
\end{equation}

\begin{lemma}\label{lemma:value-of-overline-brack-pi-circ-sigma-brack}
For all $\pi$, $\sigma$ in $S_n$, we have $\overline{\pi\circ\sigma}=\overline{\pi}\circ\overline{\sigma}^\pi.$
\end{lemma}
\begin{proof}
By definition, we have:
\begin{eqnarray*}
\overline{\pi\circ\sigma}&=&(0,1,2,\ldots,n)\circ(0,(\pi\circ\sigma)_n, (\pi\circ\sigma)_{n-1},\ldots,(\pi\circ\sigma)_1)\\
&=&(0,1,2,\ldots,n)\circ\pi\circ(0,\sigma_n, \sigma_{n-1},\ldots,\sigma_1)\circ\pi^{-1}\\
&=&(0,1,2,\ldots,n)\circ(0,\pi_n,\pi_{n-1},\ldots,\pi_1)\circ\pi\circ(0,1,2,\ldots,n)\\
&&\circ(0,\sigma_n, \sigma_{n-1},\ldots,\sigma_1)\circ\pi^{-1} \quad\quad\quad\quad\quad\quad\quad\quad\quad\quad\mbox{(using Equation~\ref{eqn:pi-from-overline-pi})}\\
&=& \overline{\pi}\circ\overline{\sigma}^\pi.
\end{eqnarray*}
\end{proof}

We are now ready to prove our main result.

\begin{theorem}\label{thm:main-theorem}
Let $S$ be a subset of $S_n$ whose elements are mapped by $f(\cdot)$ onto $S'\subseteq A_{n+1}$. Moreover, let $\mathscr C$ be the union of the conjugacy classes (of $S_{n+1}$) that intersect with $S'$; then for any $\pi$ in $S_n$, any factorisation of $\pi$ into $t$ elements of $S$ yields a factorisation of $\overline{\pi}$ into $t$ elements of $\mathscr C$.
\end{theorem}
\begin{proof}
Induction on $t$. The base case is $\pi\in S$, and clearly $\overline{\pi}\in S'\subseteq \mathscr C$. For the induction, let $\pi=g_t\circ g_{t-1}\circ\cdots\circ g_1$, where $g_i\in S$ for $1\leq i\leq t$, and let $\sigma=g_{t-1}\circ\cdots\circ g_2\circ g_1$\ ; by Lemma~\ref{lemma:value-of-overline-brack-pi-circ-sigma-brack}, we have:
$$\overline{\pi}=\overline{g_t\circ g_{t-1}\circ\cdots\circ g_2\circ g_1}  =\overline{g_t\circ\sigma}=\overline{g_t}\circ\overline{\sigma}^{g_t}.$$
By induction, $\overline{\sigma}=g'_{t-1}\circ g'_{t-2}\circ\cdots\circ g'_1$\ , where $g'_i\in \mathscr C$ for $1\leq i\leq t-1$; therefore:
\begin{eqnarray*}
g_t\circ\overline{\sigma}\circ g_t^{-1}&=&g_t\circ g'_{t-1}\circ g'_{t-2}\circ\cdots\circ g'_1\circ g_t^{-1} \\
                                       &=&\underbrace{g_t\circ g'_{t-1}\circ g_t^{-1}}_{h_{t-1}}\circ \underbrace{g_t\circ g'_{t-2}\circ g_t^{-1}}_{h_{t-2}}\circ g_t\circ\cdots\circ g_t^{-1}\circ \underbrace{g_t\circ g'_1\circ g_t^{-1}}_{h_1}\ ,
\end{eqnarray*}
and $h_1, h_2, \ldots, h_{t-1}\in\mathscr C$, which completes the proof.
\end{proof}

We will use Theorem~\ref{thm:main-theorem} in the next two sections to prove lower bounds on several edit distances between permutations.


\section{Recovering previous results}\label{sec:recovery}

We illustrate how to use Theorem~\ref{thm:main-theorem} to recover two previously known results on $bid$ and $td$. The general idea is as follows: as we explained in Section~\ref{sec:intro-gr-prefix}, if $S$ is symmetric, then any sorting sequence of length $t$ for $\pi$ made of elements of $S$ yields a factorisation of $\pi$ into the product of $t$ elements of $S$, 
which can in turn be converted, as in the proof of Theorem~\ref{thm:main-theorem}, into a factorisation of $\overline{\pi}$ into the product of $t$ elements of $S'\subseteq\mathscr C$. Therefore, the length of a shortest factorisation of $\overline{\pi}$ into the product of elements of $\mathscr C$ is a lower bound on the length of a factorisation of $\pi$ into the product of elements of $S$, and we can obtain a lower bound on the distance of interest by:
\begin{enumerate}
 \item characterising the set of images of the elements in $S$ by $f(\cdot)$, and
 \item computing the distance of $\overline{\pi}$ with respect to $\mathscr C$.
\end{enumerate}
Let us now show how we can obtain a lower bound on the block-interchange distance. We start by characterising the image of a block-interchange by our mapping.

\begin{lemma}\label{lemma:overline-beta-equals-two-crossing-two-cycles}
For any block-interchange $\beta(i,j,k,l)$ in $S_n$, we have $$\overline{\beta(i,j,k,l)}=(j,l)\circ(i,k).$$
\end{lemma}
\begin{proof}
Equation~\ref{eq-alpha} and Definition~\ref{def:block-interchange} yield
\begin{eqnarray*}
&&(0,1,2,\ldots,n)\circ(0,n,n-1,\ldots,l,j-1,j-2,\ldots,i,k-1,k-2,\ldots,\\
&&j,l-1,l-2,\ldots, k,i-1,i-2,\ldots,1) \\
&=&(0)(n)(n-1)\cdots(l+1)(l,j)(l-1)\cdots(i+1)(i,k)(i-1)\cdots(1)\\
&=&(j, l)\circ(i,k).
\end{eqnarray*}
\end{proof}

Note that $(j,l)$ and $(i,k)$ might not be disjoint, since Definition~\ref{def:block-interchange} 
allows for $j=k$ (hence the use of $\circ$ in the expression of $\overline{\beta(i,j,k,l)}$). 
We can now recover a known lower bound on the block-interchange distance, which is actually the exact distance as shown by \citet{christie-block}.

\begin{theorem}\label{thm:bid-lower-bound}
\cite{christie-block} For all $\pi$ in $S_n$, we have $bid(\pi)\geq\frac{n+1-c(\Gamma(\overline{\pi}))}{2}$.
\end{theorem}
\begin{proof}
By Theorem~\ref{thm:main-theorem} and Lemma~\ref{lemma:overline-beta-equals-two-crossing-two-cycles}, a lower bound on $bid(\pi)$ is given by the length of a minimum factorisation of $\overline{\pi}$ into the product of pairs of exchanges. Since this length equals $(n+1-c(\Gamma(\overline{\pi})))/2$ (see \eg \citet{jerrum-complexity}), the proof follows.
\end{proof}

Let us now characterise the image of a transposition by our mapping.

\begin{lemma}\label{lemma:overline-tau-equals-three-cycle}
For any transposition $\tau(i,j,l)$, we have $$\overline{\tau(i,j,l)}=(i,l,j).$$
\end{lemma}
\begin{proof}
As noted in Section~\ref{sec:intro-gr-prefix}, we have $\tau(i,j,l)=\beta(i,j,j,l)$; Lemma~\ref{lemma:overline-beta-equals-two-crossing-two-cycles} yields
$$\overline{\tau(i,j,l)}=\overline{\beta(i,j,j,l)}=(j,l)\circ(i,j)=(i,l,j).$$
\end{proof}

We can now recover the following known lower bound on the transposition distance. Recall that $c_{odd}(\Gamma(\overline{\pi}))$ denotes the number of odd cycles in $\Gamma(\overline{\pi})$.

\begin{theorem}\label{thm:bafna-pevzner-lower-bound}
\cite{bafna-transpositions} For all $\pi$ in $S_n$, we have $td(\pi)\geq\frac{n+1-c_{odd}(\Gamma(\overline{\pi}))}{2}$.
\end{theorem}
\begin{proof}
By Theorem~\ref{thm:main-theorem} and Lemma~\ref{lemma:overline-tau-equals-three-cycle}, a lower bound on $td(\pi)$ is given by the length of a minimum factorisation of $\overline{\pi}$ into the product of $3$-cycles. Since this length equals $(n+1-c_{odd}(\Gamma(\overline{\pi})))/2$ (see \eg \citet{jerrum-complexity}), the proof follows.
\end{proof}


\section{New results on the prefix transposition distance}\label{sec:sbpt}

\citet{dias-prefix} initiated the study of sorting by prefix transpositions, and derived a lower bound on the corresponding distance using the following concepts. 

\begin{definition}\label{def:ptb}
Given a permutation $\pi$ in $S_n$, build the permutation $\widetilde\pi=\langle 0\ \pi_1\ \cdots\ \pi_n$ $n+1\rangle$; a pair $(\widetilde\pi_i,\widetilde\pi_{i+1})$ with $0\leq i\leq n$ is a \emph{prefix transposition breakpoint} if $\widetilde\pi_{i+1}\neq\widetilde\pi_i+1$ or if $i=0$, and an \emph{adjacency} otherwise. 
\end{definition}

The number of prefix transposition breakpoints of $\pi$ is denoted by $ptb(\pi)$.  Noting that a prefix transposition can create at most two adjacencies and that $\iota$ is the only permutation with one prefix transposition breakpoint, \citeauthor{dias-prefix} obtained the following lower bound.

\begin{lemma}\label{lemma:prefix-transposition-breakpoint-lower-bound}
\cite{dias-prefix} For any $\pi$ in $S_n$: 
\begin{eqnarray}\label{eqn:dias-meidanis-ptd-lb}
ptd(\pi)\geq \left\lceil\frac{ptb(\pi)-1}{2}\right\rceil.
\end{eqnarray}
\end{lemma}

\citet{chitturi-bounding} later obtained another lower bound on the prefix transposition distance. They used the following concepts, based on permutations of $\{0,1,2,\ldots,n-1\}$ rather than $\{1,2,\ldots,n\}$. 

\begin{definition}
For a permutation $\pi$ of $\{0,1,2,\ldots,n-1\}$, an ordered pair $(\pi_i, \pi_{i+1})$ is an \emph{anti-adjacency} if $\pi_{i+1}=\pi_i-1\pmod{n}$. A \emph{strip} in a permutation $\pi$ is a maximal interval of $\pi$ that contains only adjacencies, and a \emph{clan} is a maximal interval of $\pi$ that contains only anti-adjacencies.  
\end{definition}

\citeauthor{chitturi-bounding} proved the following lower bound.

\begin{lemma}\label{lemma:chitturi-and-sudborough-lower-bound}
\cite{chitturi-bounding} For any permutation $\pi$ of $\{0,1,2,\ldots,n-1\}$, let $\Upsilon(\pi)$ denote the set of all clans of $\pi$ of length at least $3$, and $s(\pi)$ denote the number of strips of $\pi$. 
Then 
\begin{eqnarray}\label{eqn:chitturi-sudborough-ptd-lb}
ptd(\pi)\geq \frac{s(\pi)+\frac{\sum_{C\in\Upsilon(\pi)}\left(|C|-2\right)}{3}}{2}.
\end{eqnarray}
\end{lemma}

We will prove a new lower bound on the prefix transposition distance (Theorem~\ref{thm:my-ptd-lower-bound-i} page~\pageref{thm:my-ptd-lower-bound-i}), using our model and \citeauthor{akers-star}'s results~\cite{akers-star} on computing the prefix exchange distance:

\begin{theorem}\label{thm:formula-for-pexc}
\cite{akers-star} For any $\pi$ in $S_n$, we have
$$pexc(\pi)=
n+c(\Gamma(\pi))-2c_1(\Gamma(\pi))-\left\{
\begin{array}{ll}
0 & \mbox{if } \pi_1= 1, \\
2 & \mbox{otherwise},
\end{array}
\right.$$
where $c_1(\Gamma(\pi))$ denotes the number of $1$-cycles in $\Gamma(\pi)$, or equivalently the number of fixed points of $\pi$.
\end{theorem}

\subsection{An improved lower bound}

Using our theory, we prove a \emph{new} lower bound on $ptd(\pi)$ and show that it always outperforms~\eqref{eqn:dias-meidanis-ptd-lb}. We will find it convenient to express $ptb(\pi)$ as follows.

\begin{lemma}\label{lemma:another-expression-for-ptb}
For any $\pi$ in $S_n$, we have
$$ptb(\pi)=n+1-c_1(\Gamma(\overline{\pi}))+\left\{
\begin{array}{ll}
1 & \mbox{if } \pi_1=1, \\
0 & \mbox{otherwise}.
\end{array}
\right.$$
\end{lemma}
\begin{proof}
The formula results from the observation that, among the $n+1$ pairs of adjacent elements in $\widetilde\pi$, each adjacency in $\widetilde\pi$ gives rise to a $1$-cycle in $\Gamma(\overline{\pi})$, and from the fact that if $\pi_1=1$, then we counted the $1$-cycle that corresponds to $(0,1)$ as an adjacency, which is contrary to Definition~\ref{def:ptb} and which we correct by adding $1$.
\end{proof}

As explained in Section~\ref{sec:recovery}, we can obtain a lower bound on the prefix transposition distance by characterising the image of a prefix transposition by $f(\cdot)$ and computing the associated distance. We already know that transpositions are mapped onto $3$-cycles (see Lemma~\ref{lemma:overline-tau-equals-three-cycle} page~\pageref{lemma:overline-tau-equals-three-cycle}); in the case of prefix transpositions, it is easily seen that these $3$-cycles will always contain element $0$. Therefore, we need to be able to compute the length of a minimum factorisation of $\pi$ in $S_n$ into a product of $3$-cycles, where each $3$-cycle in the factorisation is further required to contain the first element. Let us denote the corresponding distance $d_3^1(\pi)$; the following result shows how to compute it.

\begin{lemma}\label{lemma:distance-using-prefix-cycles-of-length-three}
For any $\pi$ in $A_n$, we have
$$d^1_3(\pi)=\frac{n+c(\Gamma(\pi))}{2}-c_1(\Gamma(\pi))-\left\{
\begin{array}{ll}
0 & \mbox{if } \pi_1= 1, \\
1 & \mbox{otherwise}.
\end{array}
\right.$$
\end{lemma}
\begin{proof}
Given a minimum factorisation of length $\ell$ of an even permutation $\pi$ into prefix exchanges, we can construct a sequence of $\ell/2$
$3$-cycles by noting that $(1,j)\circ(1,i)=(1,i,j)$. Therefore $d^1_3(\pi)\leq \ell/2$. On the other hand, assume there exists a shorter sequence of $3$-cycles acting on the first element whose product is $\pi$; then one can split each of these $3$-cycles into two prefix exchanges using the relation above and find a shorter expression for $\pi$ as a product of prefix exchanges, a contradiction. The result follows from Theorem~\ref{thm:formula-for-pexc}.
\end{proof}

As a corollary, we obtain the following new lower bound on the prefix transposition distance.

\begin{theorem}\label{thm:my-ptd-lower-bound-i}
For any $\pi$ in $S_n$, we have
\begin{eqnarray}\label{eqn:my-ptd-lower-bound-i}
ptd(\pi)\geq\frac{n+1+c(\Gamma(\overline{\pi}))}{2}-c_1(\Gamma(\overline{\pi}))-\left\{
\begin{array}{ll}
0 & \mbox{if } \pi_1= 1, \\
1 & \mbox{otherwise}.
\end{array}
\right.
\end{eqnarray}
\end{theorem}
\begin{proof}
Follows immediately from Theorem~\ref{thm:main-theorem} and Lemma~\ref{lemma:distance-using-prefix-cycles-of-length-three}.
\end{proof}

An immediate question is how tight this new lower bound actually is. We will answer this question experimentally in Section~\ref{sec:ptd-experimental}, where we will see that many more permutations are tight with respect to our new result than with respect to the previously known lower bounds. We will in the meantime conclude this section by proving that our lower bound always outperforms Dias and Meidanis' (given by Lemma~\ref{lemma:prefix-transposition-breakpoint-lower-bound}).

\begin{theorem}\label{thm:tightness-of-my-lower-bound-on-ptd}
For all $\pi$ in $S_n$, the value of lower bound~\eqref{eqn:my-ptd-lower-bound-i} always exceeds that of lower bound~\eqref{eqn:dias-meidanis-ptd-lb}.
\end{theorem}
\begin{proof}
Assume $\pi\neq\iota$ (otherwise the result trivially holds); this implies that $\Gamma(\overline{\pi})$ has at least one cycle of length at least $2$, which means that $c(\Gamma(\overline{\pi}))-c_1(\Gamma(\overline{\pi}))\geq 1$. There are two cases to prove: if $\pi_1=1$, then lower bound~\eqref{eqn:dias-meidanis-ptd-lb} becomes
$$\left\lceil\frac{(n+1-c_1(\Gamma(\overline{\pi}))+1)-1}{2}\right\rceil=\left\lceil\frac{n+1-c_1(\Gamma(\overline{\pi}))}{2}\right\rceil,$$
and lower bound~\eqref{eqn:my-ptd-lower-bound-i} satisfies
\begin{eqnarray*}
\frac{n+1+c(\Gamma(\overline{\pi}))-2c_1(\Gamma(\overline{\pi}))}{2}&\geq&\frac{n+2-c_1(\Gamma(\overline{\pi}))}{2} \geq\left\lceil\frac{n+1-c_1(\Gamma(\overline{\pi}))}{2}\right\rceil.
\end{eqnarray*}
On the other hand, if $\pi_1\neq 1$, then lower bound~\eqref{eqn:dias-meidanis-ptd-lb} becomes
$$\left\lceil\frac{(n+1-c_1(\Gamma(\overline{\pi})))-1}{2}\right\rceil=\left\lceil\frac{n-c_1(\Gamma(\overline{\pi}))}{2}\right\rceil,$$
and Definition~\ref{def:even-permutation} implies that for any $\pi$ in $A_n$, we have $n\equiv c(\Gamma(\pi))\pmod{2}$. Lower bound~\eqref{eqn:my-ptd-lower-bound-i} becomes
\begin{eqnarray*}
\frac{n+1+c(\Gamma(\overline{\pi}))}{2}-c_1(\Gamma(\overline{\pi}))-1&=&\left\lceil\frac{n+1+c(\Gamma(\overline{\pi}))-2c_1(\Gamma(\overline{\pi}))-2}{2}\right\rceil \\
&\geq&\left\lceil\frac{n-c_1(\Gamma(\overline{\pi}))}{2}\right\rceil.
\end{eqnarray*}
\end{proof}


\subsection{A tighter lower bound on the prefix transposition diameter}

\citet{dias-prefix} observed that the prefix transposition diameter lies between $n/2$ and $n-1$, and conjectured that it is equal to $n-\left\lfloor\frac{n}{4}\right\rfloor$. \citet{chitturi-bounding} then improved those bounds to $2n/3$ and $n-\log_8 n$, respectively. Using our new lower bound, we further improve the lower bound on the prefix transposition diameter. We prove our result in a constructive way, by building families of permutations whose prefix transposition distance is at least $\left\lfloor 3n/4 \right \rfloor$. Figure~\ref{fig:tight-permutations-for-my-new-lb-on-ptdiameter}, which follows our result, shows examples of such permutations. The proof uses permutations from the following class, which has proved useful in the analysis of several other rearrangement problems~\cite{fertin-combinatorics}.

\begin{definition}
  A permutation $\pi$ in $S_n$ is a $2$-permutation if all cycles in $\overline{\pi}$ have length $2$.
\end{definition}

Note that the above definition requires $n\equiv 3\pmod{4}$: indeed, $n+1$ must be even in order to obtain a partition of the elements of $\overline{\pi}$ into pairs, and $(n+1)/2$ is also even by the definition of $\overline{\pi}$.

\begin{theorem}\label{thm:better-lower-bound-on-ptdiameter}
For all $n$, the prefix transposition diameter of $S_n$ is at least $\left\lfloor 3n/4\right\rfloor$.
\end{theorem}
\begin{proof}
If $n=1$ or $2$, the result is easily verified. For $n\geq 3$, we construct for each value of $n\pmod{4}$ a suitable permutation. Figure~\ref{fig:tight-permutations-for-my-new-lb-on-ptdiameter} shows an example for each case of the proof.
\begin{enumerate}
 \item if $n\equiv 3 \pmod{4}$, then any $2$-permutation $\pi$ in $S_n$ is a valid candidate: indeed, $\overline{\pi}$ contains in this case exactly $(n+1)/2$ cycles of length $2$, and Theorem~\ref{thm:my-ptd-lower-bound-i} yields
$$ptd(\pi)\geq\frac{n+1+\frac{n+1}{2}}{2}-1=\frac{3n+3-4}{4}=\frac{3n-1}{4}.$$
 \item if $n\equiv 0 \pmod{4}$, we build a permutation $\sigma$ in $S_{n}$ by inserting a new first element as a fixed point in $\overline{\pi}$, where $\pi$ is the permutation in $S_{n-1}$ constructed in the previous case. $\overline{\sigma}$ contains $n/2$ cycles of length $2$ and one cycle of length $1$ that corresponds to the fact that $\sigma_1=1$. Theorem~\ref{thm:my-ptd-lower-bound-i} then yields
$$ptd(\sigma)\geq\frac{n+1+\frac{n}{2}+1}{2}-1=\frac{2n+2+n+2-4}{4}=\frac{3n}{4}.$$
 \item if $n\equiv 1 \pmod{4}$, we build a permutation $\xi$ in $S_{n}$ by inserting a fixed point anywhere in $\overline{\sigma}$, where $\sigma$ is the permutation in $S_{n-1}$ built in the previous case. $\overline{\xi}$ contains $(n+1-2)/2$ cycles of length $2$ and two cycles of length $1$, and $\xi_1=1$. Theorem~\ref{thm:my-ptd-lower-bound-i} then yields
$$ptd(\xi)\geq\frac{n+1+\frac{n+1-2}{2}+2}{2}-2=\frac{2n+2+n+1-2+4-8}{4}=\frac{3n-3}{4}.$$
 \item if $n\equiv 2 \pmod{4}$, we build a permutation $\tau$ in $S_{n}$ by appending a $3$-cycle to any permutation $\overline{\pi}$ such that $\pi$ is a $2$-permutation in $S_{n-3}$. $\overline{\tau}$ contains $(n+1-3)/2$ cycles of length $2$ and one cycle of length $3$, and Theorem~\ref{thm:my-ptd-lower-bound-i} yields
$$ptd(\tau)\geq\frac{n+1+\frac{n+1-3}{2}+1}{2}-1=\frac{2n+2+n+1-3+2-4}{4}=\frac{3n-2}{4}.$$
\end{enumerate}
\end{proof}

\begin{figure}[htbp]
\begin{center}
\scalebox{0.75}{
\begin{tabular}{cccc}
 \begin{tikzpicture}[scale=0.75,>=stealth]
  \path (2*45:2.0cm) node [vertex,draw,circle] (v0) {};
  \path (3*45:2.0cm) node [vertex,draw,circle] (v1) {};
  \path (4*45:2.0cm) node [vertex,draw,circle] (v2) {};
  \path (5*45:2.0cm) node [vertex,draw,circle] (v3) {};
  \path (6*45:2.0cm) node [vertex,draw,circle] (v4) {};
  \path (7*45:2.0cm) node [vertex,draw,circle] (v5) {};
  \path (0*45:2.0cm) node [vertex,draw,circle] (v6) {};
  \path (1*45:2.0cm) node [vertex,draw,circle] (v7) {};
  \path (0*45:2.4cm) node {$2$};
  \path (1*45:2.4cm) node {$3$};
  \path (2*45:2.4cm) node {$0$};
  \path (3*45:2.4cm) node {$5$};
  \path (4*45:2.4cm) node {$6$};
  \path (5*45:2.4cm) node {$7$};
  \path (6*45:2.4cm) node {$4$};
  \path (7*45:2.4cm) node {$1$};
  \draw[->] (v1) -- (v2);
  \draw[->] (v2) -- (v3);
  \draw[->] (v3) -- (v4);
  \draw[->] (v4) -- (v5);
  \draw[->] (v5) -- (v6);
  \draw[->] (v6) -- (v7);
  \draw[->] (v7) -- (v0);
  \draw[->] (v0) -- (v1);
  \draw[gray,->] (v0) [out=285,in=135] to (v5);
  \draw[gray,->] (v5) [bend left=40] to (v6);
  \draw[gray,->] (v6) [bend left=40] to (v7);
  \draw[gray,->] (v7) [out=225,in=75] to (v4);
  \draw[gray,->] (v4) [out=105,in=315] to (v1);
  \draw[gray,->] (v1) [bend left=40] to (v2);
  \draw[gray,->] (v2) [bend left=40] to (v3);
  \draw[gray,->] (v3) [out=45,in=255] to (v0);
\end{tikzpicture}
&
\begin{tikzpicture}[scale=0.75,>=stealth]
  \path (2*40:2.0cm) node [vertex,draw,circle] (v0) {};
  \path (3*40:2.0cm) node [vertex,draw,circle] (v1) {};
  \path (4*40:2.0cm) node [vertex,draw,circle] (v2) {};
  \path (5*40:2.0cm) node [vertex,draw,circle] (v3) {};
  \path (6*40:2.0cm) node [vertex,draw,circle] (v4) {};
  \path (7*40:2.0cm) node [vertex,draw,circle] (v5) {};
  \path (8*40:2.0cm) node [vertex,draw,circle] (v6) {};
  \path (0*40:2.0cm) node [vertex,draw,circle] (v7) {};
  \path (1*40:2.0cm) node [vertex,draw,circle] (v8) {};
  \path (0*40:2.4cm) node {$4$};
  \path (1*40:2.4cm) node {$1$};
  \path (2*40:2.4cm) node {$0$};
  \path (3*40:2.4cm) node {$6$};
  \path (4*40:2.4cm) node {$7$};
  \path (5*40:2.4cm) node {$8$};
  \path (6*40:2.4cm) node {$5$};
  \path (7*40:2.4cm) node {$2$};
  \path (8*40:2.4cm) node {$3$};
  \draw[->] (v1) -- (v2);
  \draw[->] (v2) -- (v3);
  \draw[->] (v3) -- (v4);
  \draw[->] (v4) -- (v5);
  \draw[->] (v5) -- (v6);
  \draw[->] (v6) -- (v7);
  \draw[->] (v7) -- (v8);
  \draw[->] (v8) -- (v0);
  \draw[->] (v0) -- (v1);
  \draw[gray,->] (v0) [bend right=40] to (v8);
  \draw[gray,->] (v8) [out=220,in=100] to (v5);
  \draw[gray,->] (v5) [bend left=40] to (v6);
  \draw[gray,->] (v6) [bend left=40] to (v7);
  \draw[gray,->] (v7) [out=180,in=45] to (v4);
  \draw[gray,->] (v4) [out=70,in=290] to (v1);
  \draw[gray,->] (v1) [bend left=40] to (v2);
  \draw[gray,->] (v2) [bend left=40] to (v3);
  \draw[gray,->] (v3) [out=0,in=260] to (v0);
\end{tikzpicture}
&
\begin{tikzpicture}[scale=0.75,>=stealth]
  \path (2*36:2.0cm) node [vertex,draw,circle] (v0) {};
  \path (3*36:2.0cm) node [vertex,draw,circle] (v1) {};
  \path (4*36:2.0cm) node [vertex,draw,circle] (v2) {};
  \path (5*36:2.0cm) node [vertex,draw,circle] (v3) {};
  \path (6*36:2.0cm) node [vertex,draw,circle] (v4) {};
  \path (7*36:2.0cm) node [vertex,draw,circle] (v5) {};
  \path (8*36:2.0cm) node [vertex,draw,circle] (v6) {};
  \path (9*36:2.0cm) node [vertex,draw,circle] (v7) {};
  \path (0*36:2.0cm) node [vertex,draw,circle] (v8) {};
  \path (1*36:2.0cm) node [vertex,draw,circle] (v9) {};
  \path (0*36:2.4cm) node {$4$};
  \path (1*36:2.4cm) node {$1$};
  \path (2*36:2.4cm) node {$0$};
  \path (3*36:2.4cm) node {$9$};
  \path (4*36:2.4cm) node {$6$};
  \path (5*36:2.4cm) node {$7$};
  \path (6*36:2.4cm) node {$8$};
  \path (7*36:2.4cm) node {$5$};
  \path (8*36:2.4cm) node {$2$};
  \path (9*36:2.4cm) node {$3$};
  \draw[->] (v1) -- (v2);
  \draw[->] (v2) -- (v3);
  \draw[->] (v3) -- (v4);
  \draw[->] (v4) -- (v5);
  \draw[->] (v5) -- (v6);
  \draw[->] (v6) -- (v7);
  \draw[->] (v7) -- (v8);
  \draw[->] (v8) -- (v9);
  \draw[->] (v9) -- (v0);
  \draw[->] (v0) -- (v1);
  \draw[gray,->] (v0) [bend right=40] to (v9);
  \draw[gray,->] (v9) [out=220,in=100] to (v6);
  \draw[gray,->] (v6) [bend left=40] to (v7);
  \draw[gray,->] (v7) [bend left=40] to (v8);
  \draw[gray,->] (v8) [out=180,in=55] to (v5);
  \draw[gray,->] (v5) [out=80,in=330] to (v2);
  \draw[gray,->] (v2) [bend left=40] to (v3);
  \draw[gray,->] (v3) [bend left=40] to (v4);
  \draw[gray,->] (v4) [out=30,in=280] to (v1);
  \draw[gray,->] (v1) [bend right=40] to (v0);
\end{tikzpicture}
&
\begin{tikzpicture}[scale=0.75,>=stealth]
  \path ( 2*32.72:2.0cm) node [vertex,draw,circle] (v0) {};
  \path ( 3*32.72:2.0cm) node [vertex,draw,circle] (v1) {};
  \path ( 4*32.72:2.0cm) node [vertex,draw,circle] (v2) {};
  \path ( 5*32.72:2.0cm) node [vertex,draw,circle] (v3) {};
  \path ( 6*32.72:2.0cm) node [vertex,draw,circle] (v4) {};
  \path ( 7*32.72:2.0cm) node [vertex,draw,circle] (v5) {};
  \path ( 8*32.72:2.0cm) node [vertex,draw,circle] (v6) {};
  \path ( 9*32.72:2.0cm) node [vertex,draw,circle] (v7) {};
  \path (10*32.72:2.0cm) node [vertex,draw,circle] (v8) {};
  \path ( 0*32.72:2.0cm) node [vertex,draw,circle] (v9) {};
  \path ( 1*32.72:2.0cm) node [vertex,draw,circle] (v10) {};
  \path ( 0*32.72:2.4cm) node {$2$};
  \path ( 1*32.72:2.4cm) node {$3$};
  \path ( 2*32.72:2.4cm) node {$0$};
  \path ( 3*32.72:2.4cm) node {$9$};
  \path ( 4*32.72:2.5cm) node {$10$};
  \path ( 5*32.72:2.4cm) node {$8$};
  \path ( 6*32.72:2.4cm) node {$5$};
  \path ( 7*32.72:2.4cm) node {$6$};
  \path ( 8*32.72:2.4cm) node {$7$};
  \path ( 9*32.72:2.4cm) node {$4$};
  \path (10*32.72:2.4cm) node {$1$};
  \draw[->] (v1) -- (v2);
  \draw[->] (v2) -- (v3);
  \draw[->] (v3) -- (v4);
  \draw[->] (v4) -- (v5);
  \draw[->] (v5) -- (v6);
  \draw[->] (v6) -- (v7);
  \draw[->] (v7) -- (v8);
  \draw[->] (v8) -- (v9);
  \draw[->] (v9) -- (v10);
  \draw[->] (v10) -- (v0);
  \draw[->] (v0) -- (v1);
  \draw[gray,->] (v0) [out=270,in=135] to (v8);
  \draw[gray,->] (v8) [bend left=40] to (v9);
  \draw[gray,->] (v9) [bend left=40] to (v10);
  \draw[gray,->] (v10) [out=220,in=100] to (v7);
  \draw[gray,->] (v7) [out=125,in=10] to (v4);
  \draw[gray,->] (v4) [bend left=40] to (v5);
  \draw[gray,->] (v5) [bend left=40] to (v6);
  \draw[gray,->] (v6) [out=70,in=330] to (v3);
  \draw[gray,->] (v3) [out=0,in=280] to (v1);
  \draw[gray,->] (v1) [bend left=40] to (v2);
  \draw[gray,->] (v2) [out=310,in=230] to (v0);
\end{tikzpicture}\\
$n\equiv 3\pmod{4}$ & $n\equiv 0\pmod{4}$ & $n\equiv 1\pmod{4}$ & $n\equiv 2\pmod{4}$ \\
\end{tabular}
}
\end{center}
\caption{Cycle graphs of permutations $\pi$ in $S_n$ with $ptd(\pi)\geq\left\lfloor 3n/4\right \rfloor$, for all values of $n\bmod{4}$.}
\label{fig:tight-permutations-for-my-new-lb-on-ptdiameter}
\end{figure}
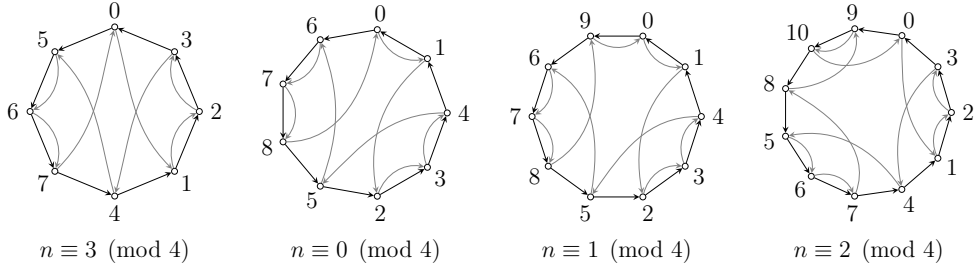

We can actually show that the lower bound on the prefix transposition distance of $2$-permutations is tight.
In order to do that, we will need the following result. We use the following relation to order black arcs:
$(\pi_i,\pi_{i-1})\prec(\pi_j,\pi_{j-1})\mbox{ if }j\geq i.$

\begin{lemma}\label{lemma:nonoriented-cycle-always-crosses-another-one}
\cite{bafna-transpositions} For any $\pi$ in $S_n$, let $C_1$ be a cycle of length $2$ in $G(\pi)$ with black arcs $a$, $b$; then there exists another cycle $C_2$ in $G(\pi)$ containing two black arcs $c$ and $d$ such that $a\prec c\prec b\prec d$ or $c\prec a\prec d\prec b$.
\end{lemma}

This result can be interpreted in a more visual way in the case of a $2$-permutation $\pi$ by saying that in $G(\pi)$, every $2$-cycle intersects with another $2$-cycle. We are now ready to prove the following result.

\begin{proposition}\label{prop:ptd-of-two-permutations}
For any $2$-permutation $\pi$ in $S_n$, we have $ptd(\pi)=(3n-1)/4$.
\end{proposition}
\begin{proof}
The lower bound has already been observed in Theorem~\ref{thm:better-lower-bound-on-ptdiameter}. To show that it is also an upper bound, we give an algorithm that sorts $\pi$ in exactly that number of steps. By Lemma~\ref{lemma:nonoriented-cycle-always-crosses-another-one}, every $2$-cycle intersects with another $2$-cycle, and as observed by \citet{bafna-transpositions}, a sequence of two transpositions on any two crossing $2$-cycles will transform them into four adjacencies:

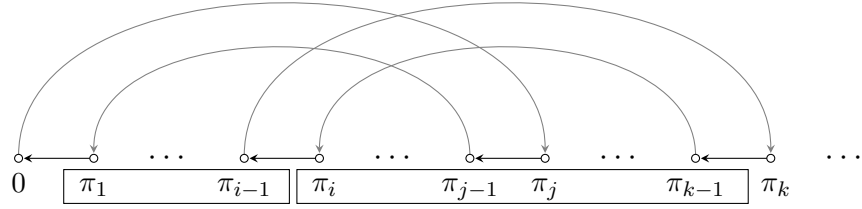
\begin{center}
 \begin{tikzpicture}[>=stealth,shorten <=1pt,shorten >=2pt]
    \foreach \p in {1,4,7,10}
        \draw[->] (\p,0) -- (\p-1,0);
    \draw[gray,->] (0,0) to  [out=90,in=90]  (7,0); 
    \draw[gray,->] (6,0) to  [out=90,in=90]  (1,0); 
    \draw[gray,->] (3,0) to  [out=90,in=90]  (10,0); 
    \draw[gray,->] (9,0) to  [out=90,in=90]  (4,0); 
    \node[vertex] at (0,0) [draw,circle] [label=below:\small $0$] {};
\begin{scope}[label distance=2pt]
    \node[vertex] at (1,0) [draw,circle] [label=below:\small $\pi_1$] {};
    \node[vertex] at (2,0) {$\cdots$};
    \node[vertex] at (3,0) [draw,circle] [label=below:\small $\pi_{i-1}$] {};
    \node[vertex] at (4,0) [draw,circle] [label=below:\small\ $\pi_{i}$] {};
    \node[vertex] at (5,0) {$\cdots$};
    \node[vertex] at (6,0) [draw,circle] [label=below:\small $\pi_{j-1}$] {};
    \node[vertex] at (7,0) [draw,circle] [label=below:\small $\pi_j$] {};
    \node[vertex] at (8,0) {$\cdots$};
    \node[vertex] at (9,0) [draw,circle] [label=below:\small $\pi_{k-1}$] {};
    \node[vertex] at (10,0) [draw,circle] [label=below:\small\ $\pi_k$] {};
    \node[vertex] at (11,0) {$\cdots$};
\end{scope}

    \draw (0.6,-0.15) rectangle (3.6,-0.6);
    \draw (3.7,-0.15) rectangle (9.7,-0.6);
\end{tikzpicture}

becomes

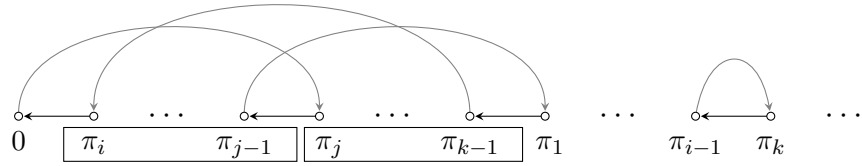
\begin{tikzpicture}[>=stealth,shorten <=1pt,shorten >=2pt]
    \foreach \p in {1,4,7,10}
        \draw[->] (\p,0) -- (\p-1,0);
    \draw[gray,->] (0,0) to  [out=90,in=90]  (4,0); 
    \draw[gray,->] (3,0) to  [out=90,in=90]  (7,0); 
    \draw[gray,->] (6,0) to  [out=90,in=90]  (1,0); 
    \draw[gray,->] (9,0) to  [controls=+(80:1) and +(100:1)]  (10,0); 
    \node[vertex] at (0,0) [draw,circle] [label=below:\small $0$] {};
\begin{scope}[label distance=2pt]
    \node[vertex] at (1,0) [draw,circle] [label=below:\small $\pi_i$] {};
    \node[vertex] at (2,0) {$\cdots$};
    \node[vertex] at (3,0) [draw,circle] [label=below:\small $\pi_{j-1}$] {};
    \node[vertex] at (4,0) [draw,circle] [label=below:\small\ \ $\pi_{j}$] {};
    \node[vertex] at (5,0) {$\cdots$};
    \node[vertex] at (6,0) [draw,circle] [label=below:\small $\pi_{k-1}$] {};
    \node[vertex] at (7,0) [draw,circle] [label=below:\small\ $\pi_1$] {};
    \node[vertex] at (8,0) {$\cdots$};
    \node[vertex] at (9,0) [draw,circle] [label=below:\small $\pi_{i-1}$] {};
    \node[vertex] at (10,0) [draw,circle] [label=below:\small $\pi_k$] {};
    \node[vertex] at (11,0) {$\cdots$};
\end{scope}
    \draw (0.6,-0.15) rectangle (3.7,-0.6);
    \draw (3.8,-0.15) rectangle (6.7,-0.6);
\end{tikzpicture}

which becomes

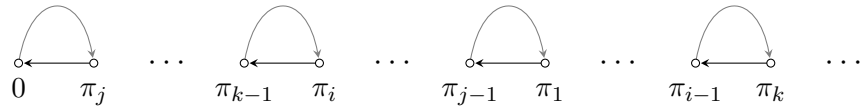
\vspace*{5mm}
\begin{tikzpicture}[>=stealth,shorten <=1pt,shorten >=2pt]
    \foreach \p in {1,4,7,10}
        \draw[->] (\p,0) -- (\p-1,0);
    \draw[gray,->] (0,0) to  [controls=+(80:1) and +(100:1)]  (1,0); 
    \draw[gray,->] (3,0) to  [controls=+(80:1) and +(100:1)]  (4,0); 
    \draw[gray,->] (6,0) to  [controls=+(80:1) and +(100:1)]  (7,0); 
    \draw[gray,->] (9,0) to  [controls=+(80:1) and +(100:1)]  (10,0); 
    \node[vertex] at (0,0) [draw,circle] [label=below:\small $0$] {};
\begin{scope}[label distance=2pt]
    \node[vertex] at (1,0) [draw,circle] [label=below:\small $\pi_j$] {};
    \node[vertex] at (2,0) {$\cdots$};
    \node[vertex] at (3,0) [draw,circle] [label=below:\small $\pi_{k-1}$] {};
    \node[vertex] at (4,0) [draw,circle] [label=below:\small\ $\pi_{i}$] {};
    \node[vertex] at (5,0) {$\cdots$};
    \node[vertex] at (6,0) [draw,circle] [label=below:\small $\pi_{j-1}$] {};
    \node[vertex] at (7,0) [draw,circle] [label=below:\small\ $\pi_1$] {};
    \node[vertex] at (8,0) {$\cdots$};
    \node[vertex] at (9,0) [draw,circle] [label=below:\small $\pi_{i-1}$] {};
    \node[vertex] at (10,0) [draw,circle] [label=below:\small $\pi_k$] {};
    \node[vertex] at (11,0) {$\cdots$};
\end{scope}
\end{tikzpicture}
\end{center} 

We transform the leftmost $2$-cycle and any $2$-cycle it crosses into four adjacencies using two prefix transpositions, which transforms $\pi$ into a permutation $\sigma$ that contains $\frac{n+1}{2}-2$ cycles of length $2$ and fixes the first element. Then, we carry out again this process until $\sigma$ is sorted, but we need three prefix transpositions at each step, since one move must be wasted to move the fixed points in $\sigma$'s prefix out of the way, for instance as follows:
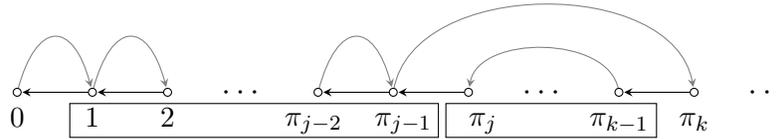
\begin{center}
 \begin{tikzpicture}[>=stealth,shorten <=1pt,shorten >=2pt]
    \foreach \p in {1,2,5,6,9}
        \draw[->] (\p,0) -- (\p-1,0);
    \draw[gray,->] (0,0) to  [controls=+(75:1) and +(105:1)]  (1,0); 
    \draw[gray,->] (1,0) to  [controls=+(75:1) and +(105:1)]  (2,0); 
    \draw[gray,->] (4,0) to  [controls=+(75:1) and +(105:1)]  (5,0); 
    \draw[gray,->] (5,0) to  [out=80,in=100]  (9,0); 
    \draw[gray,->] (8,0) to  [out=100,in=80]  (6,0); 
    \node[vertex] at (0,0) [draw,circle] [label=below:\small $0$] {};
    \node[vertex] at (1,0) [draw,circle] [label=below:\small $1$] {};
    \node[vertex] at (2,0) [draw,circle] [label=below:\small $2$] {};
    \node[vertex] at (3,0) {$\cdots$};
\begin{scope}[label distance=2pt]
    \node[vertex] at (4,0) [draw,circle] [label=below:\small$\pi_{j-2}$\ \ ] {};
    \node[vertex] at (5,0) [draw,circle] [label=below:\small\ \ $\pi_{j-1}$] {};
    \node[vertex] at (6,0) [draw,circle] [label=below:\small\ \ \ $\pi_{j}$] {};
    \node[vertex] at (7,0) {$\cdots$};
    \node[vertex] at (8,0) [draw,circle] [label=below:\small $\pi_{k-1}$] {};
    \node[vertex] at (9,0) [draw,circle] [label=below:\small $\pi_{k}$] {};
    \node[vertex] at (10,0) {$\cdots$};
\end{scope}
    \draw (0.7,-0.15) rectangle (5.6,-0.6);
    \draw (5.7,-0.15) rectangle (8.5,-0.6);
\end{tikzpicture}

becomes

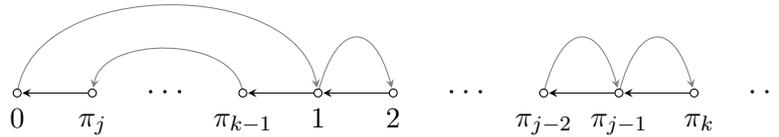
\begin{tikzpicture}[>=stealth,shorten <=1pt,shorten >=2pt]
    \foreach \p in {1,4,5,8,9}
        \draw[->] (\p,0) -- (\p-1,0);
    \draw[gray,->] (0,0) to  [out=80,in=100]  (4,0); 
    \draw[gray,->] (3,0) to  [out=100,in=80]  (1,0); 
    \draw[gray,->] (4,0) to  [controls=+(75:1) and +(105:1)]  (5,0); 
    \draw[gray,->] (7,0) to  [controls=+(75:1) and +(105:1)]  (8,0); 
    \draw[gray,->] (8,0) to  [controls=+(75:1) and +(105:1)]  (9,0); 
    \node[vertex] at (0,0) [draw,circle] [label=below:\small $0$] {};
\begin{scope}[label distance=2pt]
    \node[vertex] at (1,0) [draw,circle] [label=below:\small $\pi_j$] {};
    \node[vertex] at (2,0) {$\cdots$};
    \node[vertex] at (3,0) [draw,circle] [label=below:\small $\pi_{k-1}$] {};
\end{scope}
    \node[vertex] at (4,0) [draw,circle] [label=below:\small $1$] {};
    \node[vertex] at (5,0) [draw,circle] [label=below:\small $2$] {};
    \node[vertex] at (6,0) {$\cdots$};
\begin{scope}[label distance=2pt]
    \node[vertex] at (7,0) [draw,circle] [label=below:\small $\pi_{j-2}$] {};
    \node[vertex] at (8,0) [draw,circle] [label=below:\small $\pi_{j-1}$] {};
    \node[vertex] at (9,0) [draw,circle] [label=below:\small\ $\pi_k$] {};
    \node[vertex] at (10,0) {$\cdots$};
\end{scope}
\end{tikzpicture}
\end{center} 
The algorithm is guaranteed to terminate, since after applying each sequence of three transpositions of the form described above, we obtain either $\iota$, or a permutation on which we can repeat the same process by Lemma~\ref{lemma:nonoriented-cycle-always-crosses-another-one}. 
The proof follows from the fact that the number of prefix transpositions used by this algorithm is $$2+\frac{3}{2}\left(\frac{n+1}{2}-2\right)=\frac{8+3n-9}{4}=\frac{3n-1}{4}.$$
\end{proof}


\section{Experimental results}\label{sec:ptd-experimental}

We generated all permutations in $S_n$, for $1\leq n\leq 12$, along with their prefix transposition distance, and compared lower bounds~\eqref{eqn:dias-meidanis-ptd-lb}, \eqref{eqn:chitturi-sudborough-ptd-lb} and~\eqref{eqn:my-ptd-lower-bound-i} to the actual distance. Table~\ref{tab:comparing-lower-bounds-on-ptd} shows the results. It can be observed that many more permutations are tight with respect to our lower bound (column~5) than with respect to \citeauthor{dias-prefix}' (column~3) or \citeauthor{chitturi-bounding}'s (column~4).

\begin{table}[htbp]
\centering
{
\begin{tabular}{r|r|r|r|r}
$n$ &\        $n!$ &\ tight w.r.t.~\eqref{eqn:dias-meidanis-ptd-lb} &\ tight w.r.t.~\eqref{eqn:chitturi-sudborough-ptd-lb} &\ tight w.r.t.~\eqref{eqn:my-ptd-lower-bound-i} \\
\hline                                               
  1 &             1 &            1 &            1 &             1  \\
  2 &             2 &            2 &            2 &             2  \\
  3 &             6 &            4 &            4 &             6  \\
  4 &            24 &           13 &           15 &            22  \\
  5 &           120 &           41 &           48 &           106  \\
  6 &           720 &          196 &          255 &           574  \\
  7 &        5\,040 &          862 &       1\,144 &        3\,782  \\
  8 &       40\,320 &       5\,489 &       7\,737 &       27\,471  \\
  9 &      362\,880 &      31\,033 &      44\,187 &      229\,167  \\
 10 &   3\,628\,800 &     247\,006 &     369\,979 &   2\,103\,510  \\
 11 &  39\,916\,800 &  1\,706\,816 &  2\,575\,693 &  21\,280\,564  \\
 12 & 479\,001\,600 & 16\,302\,397 & 25\,791\,862 & 236\,651\,919  \\
\end{tabular}
}
\caption{Comparison of all known lower bounds on the prefix transposition distance. Column $3$ lists the number of cases where \citeauthor{dias-prefix}'s lower bound is tight~\cite[page 48]{fortuna-distancias}, column $4$ lists the number of cases where \citeauthor{chitturi-bounding}'s lower bound is tight, and column~$5$ lists the number of cases where our lower bound is tight.}
\label{tab:comparing-lower-bounds-on-ptd}
\end{table}

We also examined how large the gap between our lower bound and the actual prefix transposition distance can get. Table~\ref{tab:gap-between-my-ptd-lower-bound-and-ptd} counts permutations whose prefix transposition distance equals our lower bound plus $\Delta$. We note that, for $n\leq 9$, all permutations have a prefix transposition distance that is at most our lower bound plus $2$ (plus $3$ for $n\leq 12$).

\begin{table}[htbp]
\centering
{
\begin{tabular}{r|r|r|r|r|r}
$n$ &\        $n!$ & \ $\Delta=0$ &\  $\Delta=1$ &\ $\Delta=2$ &\ $\Delta=3$ \\
\hline
  1 &             1 &             1 &             0 &            0 &      0 \\
  2 &             2 &             2 &             0 &            0 &      0 \\
  3 &             6 &             6 &             0 &            0 &      0 \\
  4 &            24 &            22 &             2 &            0 &      0 \\
  5 &           120 &           106 &            14 &            0 &      0 \\
  6 &           720 &           574 &           143 &            3 &      0 \\
  7 &        5\,040 &        3\,782 &        1\,234 &           24 &      0 \\
  8 &       40\,320 &       27\,471 &       12\,310 &          539 &      0 \\
  9 &      362\,880 &      229\,167 &      128\,576 &       5\,137 &      0 \\
 10 &   3\,628\,800 &   2\,103\,510 &   1\,427\,966 &      97\,321 &      3 \\
 11 &  39\,916\,800 &  21\,280\,564 &  17\,532\,948 &  1\,103\,254 &     34 \\
 12 & 479\,001\,600 & 236\,651\,919 & 221\,680\,237 & 20\,667\,140 & 2\,304 \\
\end{tabular}
}
\caption{Number of cases where our lower bound underestimates $ptd(\pi)$ by $\Delta$.}
\label{tab:gap-between-my-ptd-lower-bound-and-ptd}
\end{table}


\section{Further observations on $\overline{\pi}$}

Now that we have an alternate representation of the cycle graph of a permutation as another permutation, we would like to examine whether or not other results can be obtained that could be helpful in getting insight on problems related to length-constrained factorisations of permutations. We investigate in this section a few relations between $\pi$ and $\overline{\pi}$, starting with relations between the cycle structures of both permutations when subjected to particular operations. 

\subsection{Cycle structures}\label{sec:cycle-structures}

A natural question is whether conjugacy classes are preserved by $f(\cdot)$, \ie whether $\overline{\pi}$ and $\overline{\pi^\sigma}$ are in the same conjugacy class for any choice of $\pi$ and $\sigma$ in $S_n$. The answer is negative in general, as the following counter-example shows: $\pi=(1,3)(2)$ and $\tau=(1,2)(3)$ are conjugate, but $\overline{\pi}=(0,1,2,3)\circ(0,1,2,3)=(0,2)(1,3)$ and $\overline{\tau}=(0,1,2,3)\circ(0,3,1,2)=(0)(3,2,1)$ are not. 
However, the relation we are interested in holds for two particular cases, whose significance we explain below.

The following result is similar in spirit to \citeauthor{tannier-subquadratic}'s characterisation of ``inverse breakpoint graphs'' of signed permutations~\cite{tannier-subquadratic}, and shows that the cycle graphs of a permutation and of its inverse have exactly the same cycle structure.

\begin{lemma}\label{lemma:overline-pi-and-overline-inverse-pi-same-class}
 For any $\pi$ in $S_n$, we have $\overline{\pi^{-1}}=(\overline{\pi}^{-1})^{(\pi^{-1})}$.
\end{lemma}
\begin{proof}
Straightforward:
\begin{eqnarray*}
\overline{\pi^{-1}}&=&(0,1,2, \ldots, n)\circ(0, \pi^{-1}_n, \pi^{-1}_{n-1}, \ldots, \pi^{-1}_1) \\
&=&\pi^{-1}\circ\pi\circ(0,1,2, \ldots, n)\circ\pi^{-1}\circ(0, n, {n-1}, \ldots, 1)\circ\pi \\
&=&\pi^{-1}\circ(0,\pi_1,\pi_2, \ldots, \pi_n)\circ(0, n, {n-1}, \ldots, 1)\circ\pi \\
&=&\pi^{-1}\circ\overline{\pi}^{-1}\circ\pi\\
&=&(\overline{\pi}^{-1})^{(\pi^{-1})}.
\end{eqnarray*}
\end{proof}

\citeauthor{tannier-subquadratic}'s idea of examining how the cycle graph of $\pi^{-1}$ evolves when applying a signed reversal to a permutation $\pi$ -- which reverses and flips the signs of the elements of an interval of $\pi$ -- was a key point in their successful attempt at designing an algorithm with an improved running time for sorting permutations by signed reversals. The above relation allows us to derive a simple description of the more general situation (albeit restricted to ``traditional'', unsigned permutations), \ie how $\overline{\pi^{-1}}$ changes when an arbitrary rearrangement $\sigma$ is applied to $\pi$.

\begin{corollary}
For all $\pi$, $\sigma$ in $S_n$, we have $\overline{(\pi\circ\sigma)^{-1}}=(\overline{\sigma}^{-1}\circ\overline{\pi^{-1}})^{\sigma^{-1}}.$
\end{corollary}
\begin{proof}
 Lemma~\ref{lemma:value-of-overline-brack-pi-circ-sigma-brack} yields
\begin{eqnarray*}
\overline{(\pi\circ\sigma)^{-1}}=\overline{\sigma^{-1}\circ\pi^{-1}}&=&\overline{\sigma^{-1}}\circ\overline{\pi^{-1}}^{(\sigma^{-1})}\\
&=&\sigma^{-1}\circ\overline{\sigma}^{-1}\circ\sigma\circ\sigma^{-1}\circ\overline{\pi^{-1}}\circ\sigma\quad\mbox{(using Lemma~\ref{lemma:overline-pi-and-overline-inverse-pi-same-class})}\\
&=&(\overline{\sigma}^{-1}\circ\overline{\pi^{-1}})^{\sigma^{-1}}.
\end{eqnarray*}
\end{proof}

A second particular case of conjugate permutations whose transformation by $f(\cdot)$ yields two conjugate permutations is presented below. Recall that $\chi$ is the reverse permutation, \ie  $\chi=\left\langle n\ n-1\ \cdots\ 1\right\rangle$.

\begin{observation}
\label{lemma:overline-pi-and-overline-conjugate-by-chi-same-class}
 For any $\pi$ in $S_n$, we have $\overline{\pi^{\chi}}=((\overline{\pi}^{-1})^\chi)^{(0, 1, 2, \ldots, n)}$.
\end{observation}
\begin{proof}
We have by definition:
\begin{eqnarray*}
\overline{\pi^\chi}&=& (0, 1, 2, \ldots, n)\circ(0,\pi^\chi_n,\pi^\chi_{n-1}, \ldots, \pi^\chi_1) \\
&=& (0, 1, 2, \ldots, n)\circ\pi^\chi\circ(0,n,{n-1}, \ldots, 1)\circ(\pi^\chi)^{-1} \\
&=& (\pi^\chi\circ(0,n,{n-1}, \ldots, 1)\circ(\pi^\chi)^{-1}\circ(0, 1, 2, \ldots, n))^{(0, 1, 2, \ldots, n)} \\
&=& (\chi\circ\pi\circ(0,1,2,\ldots, n)\circ\pi^{-1}\circ\chi\circ(0, 1, 2, \ldots, n))^{(0, 1, 2, \ldots, n)} \\
&=& (\chi\circ\pi\circ(0,1,2,\ldots, n)\circ\pi^{-1}\circ(0,n,{n-1}, \ldots, 1)\circ\chi)^{(0, 1, 2, \ldots, n)} \\
&=& (\chi\circ\overline{\pi}^{-1}\circ\chi)^{(0, 1, 2, \ldots, n)}.
\end{eqnarray*}
\end{proof}

Conjugating $\pi$ by $\chi$ corresponds to computing its reverse complement: indeed, $\pi\circ\chi=\left\langle \pi_n\ \pi_{n-1}\ \cdots\ \pi_1\right\rangle$, and $\chi\circ(\pi\circ\chi)=\left\langle n+1-\pi_n\ n+1-\pi_{n-1}\ \cdots\ n+1-\pi_1\right\rangle$. By definition, $\pi$ and $\pi^\chi$ have the same cycle structure, and by the above result, so do their images by $f(\cdot)$. The reverse complement operation is interesting because most (but not all, prefix distances being notable exceptions~\cite{labarre-combinatorial}) genome rearrangement distances are, in addition to being left-invariant, also ``reverse complement-invariant'': for all $\pi$ and $\sigma$ in $S_n$, we have $d(\pi,\sigma)=d(\pi^\chi,\sigma^\chi)$. As a consequence, bounds obtained on the distance between $\pi$ and $\iota$ with respect to a certain set of operations can sometimes be improved by examining $\pi^{-1}$ or $\pi^\chi$.

\citet{eriksson-bridge} introduced another important equivalence relation on permutations that does not preserve their cycle structure in the classical sense but that does preserve the cycle structure of their cycle graphs. This equivalence relation, whose equivalence classes are called \emph{toric permutations}, proved useful in improving bounds on the transposition distance~\cite{eriksson-bridge,labarre-new}. We will see below that $\overline{\pi}$ provides a simple way of navigating through all cycle graphs of the permutations in the same equivalence class. The equivalence relation uses the following notion.

\begin{definition}\label{def:circular-permutation}
The \emph{circular permutation} obtained from a permutation $\pi$ in $S_n$ is $\pi^\circ=0\ \pi_1\ \pi_2\ \cdots\ \pi_n$, with indices taken modulo $n+1$ so that $0=\pi^\circ_0=\pi^\circ_{n+1}.$
\end{definition}

This circular permutation can be read starting from any position, and the original ``linear'' permutation is reconstructed by taking the element following $0$ as $\pi_1$ and removing $0$. For $x$ in $\{0, 1, 2, \ldots, n\}$, let $\overline{x}^m = (x+m)\pmod{n+1}$, and define the following operation on circular permutations:
$$m+\pi^\circ=\overline{0}^m\ \overline{\pi_1}^m\ \overline{\pi_2}^m\ \cdots\ \overline{\pi_n}^m.$$

\begin{definition}\label{def:toric-permutation}
For any $\pi$ in $S_n$, the \emph{toric permutation} $\pi^\circ_\circ$ is the set of permutations in $S_n$ reconstructed from all circular permutations $m+\pi^\circ$ with $0\leq m\leq n$.
\end{definition}

\begin{definition}\label{def:toric-equivalence}
Two permutations $\pi$, $\sigma$ in $S_n$ are \emph{torically equivalent} if $\sigma\in\pi^\circ_\circ$ (or $\pi\in\sigma^\circ_\circ$), which we also write as $\pi\equiv^\circ_\circ\sigma$.
\end{definition}

Let us illustrate those notions using our running example $\pi=\langle 4\ 1\ 6\ 2\ 5\ 7\ 3\rangle$; we have $\pi^{\circ}=0\ 4\ 1\ 6\ 2\ 5\ 7\ 3$, and
$$
\begin{array}{lll}
0+\pi^{\circ} &=& 0\ 4\ 1\ 6\ 2\ 5\ 7\ 3\\
1+\pi^{\circ} &=& 1\ 5\ 2\ 7\ 3\ 6\ 0\ 4\\
2+\pi^{\circ} &=& 2\ 6\ 3\ 0\ 4\ 7\ 1\ 5\\
3+\pi^{\circ} &=& 3\ 7\ 4\ 1\ 5\ 0\ 2\ 6\\
4+\pi^{\circ} &=& 4\ 0\ 5\ 2\ 6\ 1\ 3\ 7\\
5+\pi^{\circ} &=& 5\ 1\ 6\ 3\ 7\ 2\ 4\ 0\\
6+\pi^{\circ} &=& 6\ 2\ 7\ 4\ 0\ 3\ 5\ 1\\
7+\pi^{\circ} &=& 7\ 3\ 0\ 5\ 1\ 4\ 6\ 2\\
\end{array}
$$
which yields $\pi_{\circ}^{\circ}=\{
\left\langle 4\ 1\ 6\ 2\ 5\ 7\ 3\right\rangle$,
$\left\langle 4\ 1\ 5\ 2\ 7\ 3\ 6\right\rangle$,
$\left\langle 4\ 7\ 1\ 5\ 2\ 6\ 3\right\rangle$,
$\left\langle 2\ 6\ 3\ 7\ 4\ 1\ 5\right\rangle$,
$\left\langle 5\ 2\ 6\ 1\ 3\ 7\ 4\right\rangle$,
$\left\langle 5\ 1\ 6\ 3\ 7\ 2\ 4\right\rangle$,
$\left\langle 3\ 5\ 1\ 6\ 2\ 7\ 4\right\rangle$,
$\left\langle 5\ 1\ 4\ 6\ 2\ 7\ 3\right\rangle\}$. 
\citet{hultman-toric} proved the following interesting result.

\begin{lemma}\label{lemma:torism-preserves-cycle-graph}
\cite{hultman-toric} For all $\pi$ in $S_n$ and $0\leq m\leq n$: every cycle in $G(\pi)$ is mapped onto a cycle in $G(\sigma)$, where $\sigma$ is the permutation obtained from $m+\pi^\circ$.
\end{lemma}

In other words, if $\pi\equiv_\circ^\circ\sigma$, then $\overline{\pi}$ and $\overline{\sigma}$ are conjugate. We show below how one can iterate over the cycle graphs of all elements in $\pi^\circ_\circ$.

\begin{lemma}
 For all $\pi$, $\sigma$ in $S_n$: if $\sigma^\circ=m+\pi^\circ$, then $\overline{\sigma}=\overline{\pi}^{(0,1,2,\ldots,n)^m}$.
\end{lemma}
\begin{proof}
By Equation~\ref{eq-alpha}, we have:
\begin{eqnarray*}
 \overline{\sigma}&=&(0,1,2,\ldots,n)\circ(\sigma_0,\sigma_n,\sigma_{n-1}, \ldots,\sigma_1)\\
 &=&(0,1,2,\ldots,n)\circ(m+\pi_0,m+\pi_n,m+\pi_{n-1}, \ldots,m+\pi_1),
\end{eqnarray*}
since by hypothesis $\sigma^\circ=m+\pi^\circ$. 

On the other hand, the mapping $(\pi_0,\pi_n,\pi_{n-1}, \ldots,\pi_1)\mapsto (m+\pi_0,m+\pi_n,m+\pi_{n-1}, \ldots,m+\pi_1)$ consists in replacing each element of the cycle with its value plus $m\pmod{n+1}$, which is by definition equivalent to conjugating $(\pi_0,\pi_n,\pi_{n-1}, \ldots,\pi_1)$ by $(0,1,2,\ldots,n)^m$.
\end{proof}
\begin{corollary}
 For all $\pi$ in $S_n$, we have $\{\overline{\sigma}\ |\ \sigma\in\pi^\circ_\circ\}=\{\overline{\pi}^{(0,1,2,\ldots,n)^m}\ |\ 0\leq m\leq n\}$.
\end{corollary}

Other relations between the cycle structure of $\pi$ and that of $\overline{\pi}$ can easily be derived from previous work. The following relation allows us to bound the number of odd cycles of $\overline{\pi}$.

\begin{theorem}\label{thm:upper-bound-td}
 \cite{labarre-new} For all $\pi$ in $S_n$, we have $td(\pi)\leq n-c_{odd}(\Gamma(\pi))$.
\end{theorem}

The following result is an immediate corollary of Theorems~\ref{thm:bafna-pevzner-lower-bound} and~\ref{thm:upper-bound-td}.

\begin{corollary}
 For all $\pi$ in $S_n$, we have $2c_{odd}(\Gamma(\pi))\leq n-1+c_{odd}(\Gamma(\overline{\pi}))$.
\end{corollary}

Similarly, the following result is an immediate corollary of Theorem~\ref{thm:bid-lower-bound} and of the characterisation of exchanges as restricted block-interchanges.

\begin{corollary}
 For all $\pi$ in $S_n$, we have $2c(\Gamma(\pi))\leq n-1+c(\Gamma(\overline{\pi}))$.
\end{corollary}

\subsection{Descents of $\pi$ and cycles of $\overline{\pi}$}

Aside from relations between cycle structures, we can also establish relations between pairs of elements of $\pi$ and cycles of $\overline{\pi}$. An example of such a relation is the fact that the number of adjacencies in $\langle 0\ \pi_1\ \pi_2\ \cdots\ \pi_n\ n+1\rangle$ equals $c_1(\Gamma(\overline{\pi}))$. We will prove that a less obvious relation connects the \emph{descents} of $\pi$ (defined below) and the cycles of $\overline{\pi}$.

\begin{definition}
A \emph{descent} in a permutation $\pi$ is a pair $(\pi_{i-1},\pi_i)$ such that $\pi_i<\pi_{i-1}$.
\end{definition}

For instance, the permutation $\langle 4\downarrow 1\ 6\downarrow 2\ 5\ 7\downarrow 3\rangle$ has three descents, indicated by vertical arrows.

\begin{definition}
A cycle $C$ in $G(\pi)$ \emph{contains} a descent $(\pi_{i-1},\pi_i)$ if $(\pi_i,\pi_{i-1})$ is a black arc of $C$.
\end{definition}

We now derive bounds on the number of descents contained by cycles in $G(\pi)$.

\begin{lemma}\label{lemma:descents-and-cycles}
 For all $\pi$ in $S_n$, every cycle of length $\ell\geq 2$ in $G(\pi)$ contains at most $\ell-1$ descents and at least one descent of $\pi$.
\end{lemma}
\begin{proof}
For clarity, let us write the vertices of $C$ in the order in which $C$ visits them, starting with the element whose position in $\pi$ is maximal: we get $C=(\pi_{i_1},\pi_{j_1}, \pi_{i_2}, \pi_{j_2}, \ldots, \pi_{i_k}, \pi_{j_k})$, where $i_1$ (\resp $j_k$) is the largest (\resp smallest) position of an element of $\widetilde{\pi}$ appearing in $C$. We identify here $\widetilde{\pi}_0$ and $\widetilde{\pi}_{n+1}\equiv\widetilde{\pi}_0\pmod{n+1}$. Recall that $(\pi_{i_x}, \pi_{j_x})$ for $1\leq x\leq k$ is a black arc of $C$, and that by Definition~\ref{def:cycle-graph}, the following relation holds:
\begin{equation}
 \pi_{i_x}=\pi_{j_{x-1}}+1 \mbox{ for all } 1\leq x \leq k, \mbox{ and } \pi_{i_1}=\pi_{j_k}+1.\label{eqn:cycle-graph-conditions}
\end{equation}
\begin{enumerate}
 \item for the upper bound: assume on the contrary that $C$ contains $\ell$ descents; then every black edge of $C$ corresponds to a descent, and we have:
\begin{equation}
 \pi_{j_x}>\pi_{i_x}  \mbox{ for } 1\leq x\leq k.\label{eqn:relations-elements-cycle-all-descents}
\end{equation}
By alternating between the conditions specified by Equations~\ref{eqn:relations-elements-cycle-all-descents} and~\ref{eqn:cycle-graph-conditions}, we obtain:
\begin{eqnarray*}
 \pi_{j_k}&>&\pi_{i_k}=\pi_{j_{k-1}}+1>\pi_{i_{k-1}}+1=\pi_{j_{k-2}}+2>\cdots=\pi_{j_1}+k-1\\
&>&\pi_{i_1}+k-1=\pi_{j_k}+k,
\end{eqnarray*}
which is clearly a contradiction.

\item for the lower bound: assume on the contrary that $C$ contains no descent; we have:
\begin{equation}
 \pi_{j_x}<\pi_{i_x}  \mbox{ for } 1\leq x\leq k.\label{eqn:relations-elements-cycle}
\end{equation}
 By alternating between the conditions specified by Equations~\ref{eqn:relations-elements-cycle} and~\ref{eqn:cycle-graph-conditions}, we obtain:
\begin{eqnarray*}
 \pi_{i_1}-1=\pi_{j_k}&<&\pi_{i_k}=\pi_{j_{k-1}}+1<\pi_{i_{k-1}}+1=\pi_{j_{k-2}}+2<\cdots=\pi_{j_1}+k-1\\
&<&\pi_{i_1}+k-1.
\end{eqnarray*}
For the above relations to hold, elements from the set $A=\{\pi_{i_k}, \pi_{i_{k-1}}+1, \ldots,\pi_{i_2}+k-2\}$ can only be assigned values from the set $B=\{\pi_{i_1}+1,\pi_{i_1}+2,\ldots,\pi_{i_1}+k-2\}$. However, we have $k-1=|A|>|B|=k-2$, which clearly makes it impossible to obtain a permutation.
\end{enumerate}
Finally, note that $\pi$ and $\widetilde{\pi}$ can be regarded as equivalent as far as descents are concerned, since $(\widetilde{\pi}_0,\widetilde{\pi}_1)$ and $(\widetilde{\pi}_n,\widetilde{\pi}_{n+1})$ cannot be descents.
\end{proof}

The following result is a direct corollary of the above.

\begin{proposition}\label{prop:descents-of-twoperms}
For any $2$-permutation $\pi$ in $S_n$, we have $des(\pi)=(n+1)/2$.
\end{proposition}
\begin{proof}
By definition, $\overline{\pi}$ contains exactly $(n+1)/2$ cycles of length $2$, and by Lemma~\ref{lemma:descents-and-cycles}, each of these cycles contains exactly one descent of $\pi$.
\end{proof}


\section{Conclusions}

We presented a new framework for reformulating any edit distance problem on permutations as a minimum-length factorisation problem on a related even permutation, under the implicit assumption that the edit operations are revertible. This approach is based on a new representation of a structure known as the \emph{cycle graph}, which pervades the field of genome rearrangements in several different forms; it previously allowed us to enumerate permutations whose cycle graph decomposes into a given number of alternating cycles~\cite{doignon-hultman}, and allowed us in this work to recover two previously known results in a simple and unified way. Moreover, we used our approach to derive a new lower bound on the prefix transposition distance that, as we showed both theoretically and experimentally, is a significant improvement over previous results. From that result, we deduced an improved lower bound on the prefix transposition diameter of the symmetric group, whose exact value is still unknown. Finally, we investigated other relations between permutations and their cycle graphs that we hope will prove useful in obtaining new results.

Several interesting questions and leads for future work arise. First, our method provides an automated way of obtaining lower bounds on distances between permutations; is there an analogous way of obtaining \emph{upper} bounds instead? Second, we initiated the study of relations between statistics on a permutation and statistics on the permutation that corresponds to its cycle graph. Can other relations be deduced and used to prove other results, including tighter bounds on the distances of interest? Third, permutations are but one structure for which the cycle graph has been defined. Other structures, such as \emph{signed} permutations, give rise to a more general structure known as the \emph{breakpoint graph}. Are there analogs, or generalisations of $f(\cdot)$ that can yield similar results on signed permutations as well? 
Finally, another question is whether Cayley graphs obtained from genome rearrangement operations can yield good interconnection networks. For instance, (signed) reversals generalise the operations that generate the (burnt) pancake network, exchanges generalise the operations that generate the star network, and prefix transpositions generalise the operations that generate the bi-rotator graphs (see \citet{lakshmivarahan-symmetry} for definitions). It seems likely that collaborations between researchers in both fields could be fruitful in investigating this topic.

\section*{Acknowledgements}

The author wishes to thank Jean-Paul Doignon for suggesting the new definition of $\overline{\pi}$, which greatly simplifies formulas, as well as \citeauthor{galvao-db}, whose freely available source code~\cite{galvao-db} allowed the extension of Tables~\ref{tab:comparing-lower-bounds-on-ptd} and~\ref{tab:gap-between-my-ptd-lower-bound-and-ptd} by two lines.

\bibliographystyle{mynatstyle}
\bibliography{lowerbounds}

\end{document}

%% file: newtheorems-en.tex
\usepackage{amsthm} 

\newtheorem{corollary}{Corollary}[section]

\theoremstyle{definition} 
\newtheorem{definition}{Definition}[section]
\theoremstyle{plain} 

\newtheorem{lemma}{Lemma}[section]
\newtheorem{observation}{Observation}[section]
\newtheorem{proposition}{Proposition}[section]

\newtheorem{theorem}{Theorem}[section]